\documentclass{llncs}
\usepackage{graphicx}

\usepackage[tbtags]{amsmath} 

\usepackage{amsfonts}
\usepackage{amssymb}
\usepackage{verbatim}

\def\E{\mathbb{E}}

\def\eps{\epsilon}

\def\1{\mathbf{1}}

\def\ER{Erd\H{o}s-R\'{e}nyi }

\newtheorem{thm}{Theorem}

\newtheorem{prop}{Proposition}

\begin{document}

\title{On the Resilience of Bipartite Networks}

\author{
Shelby Heinecke\inst{1}
Will Perkins\inst{2}\and Lev Reyzin\inst{1}}

\institute{Department of Mathematics, Statistics, and Computer Science\\ University of Illinois at Chicago\\
   Chicago, IL, United States \\
 \email{\{sheine4,lreyzin\}@uic.edu} \vskip .1in
 \and
 School of Mathematics\\ University of Birmingham\\ Birmingham, England\\
 \email{math@willperkins.org}
}

\maketitle

\begin{abstract}
Motivated by problems modeling the spread
of infections in networks,
in this paper we explore which bipartite 
graphs are most resilient to widespread infections under various
parameter settings.
Namely, we study bipartite networks 
with a requirement of a minimum degree $d$
on one side under
an independent infection, independent transmission
 model.
We completely characterize the optimal graphs in the case $d=1$,
which already produces 
non-trivial behavior, and we give
 extremal results for the more general cases. We show that in the case $d=2$, surprisingly, the optimally resilient set of graphs includes a graph that is not one of the two ``extremes" found in the case $d=1$. 
 
 Then, we briefly examine the case where we force a connectivity requirement instead of a one-sided degree requirement and again, we find that the set of the most resilient graphs contains more than the two ``extremes."
 We also show that determining the subgraph of an arbitrary bipartite graph most resilient to infection is NP-hard for any one-sided minimal degree $d \ge 1$.  
\end{abstract}


\section{Introduction}

The goal of our work is to study the resilience of bipartite networks to the spread of diseases, viruses, or other contagion.
In our case, the bipartite networks will 
represent an interaction between two types of agents.
Examples of such networks 
include clients and servers or persons
and drinking wells.  In the former, one
may need to connect
clients to servers in order to minimize
the propagation of computer viruses; in the latter,
one may want to direct people to drinking
wells as to minimize the spread of infections.

Our main motivation, however, comes from the study of the 
spread of sexually transmitted diseases in heterosexual
contact networks.  This problem has been studied in the economics community, with
the assumption that each gender has some (possibly asymmetric) partner distribution.
An influential paper in economics \cite{kremer1996integrating} shows that in a mean-field model of HIV infection, strategic behavior on the part of individuals can lead to two extreme equilibria, one in which all individuals have a moderate number of partners and one in which some individuals have very few partners and other individuals have very many partners.  We study the same problem in the setting of finite networks.

Namely, the model we employ has been used by Blume~et~al.~\cite{blume2011network,blume2011networks} to study the network resilience 
problem in uniform-degree graphs.  In a variant of this model, vertices represent agents in the network
and edges represent pairwise interactions among the agents.  Each agent has an independent probability
of being initially infected and can further infect neighboring agents with some probability (see Section~\ref{sec:Model} for details). 
 
Moreover, to correspond to the motivation above, we require the interaction graph to be bipartite as well as
have minimum degree on one side of the bipartition.  This is a weaker restriction than that of Blume~et~al.~\cite{blume2011networks} and
allows for a larger class of graphs.

We study extremal and computational aspects of the model. Among our results, we show the following:
\begin{itemize}
\item We extend the analysis of the susceptibility of networks to infection to the bipartite case, motivated by problems in which there are two types of agents, such as computer terminals/servers, human sexual networks, and maps of shared resources.
\item We show that the objective function, the expected fraction of infected individuals in the network, corresponds for specific choices of parameters to the expectation of natural functions under independent edge percolation, a widely studied model in probability and combinatorics.
\item We characterize optimal graphs when one side of the bipartition has uniform degree 1 and for higher degree give optimal graphs for extremal choices of parameters. (Theorems \ref{d1thm} and \ref{d2thm}).
\item We show that the two optimally resilient ``extreme" graphs in the $d=1$ case are not sufficient in the $d=2$ case (Theorem \ref{d3thm}).
\item We show that if we instead force a connectivity requirement in lieu of a one-sided degree requirement, we again find that the two obvious ``extremes" are not sufficient.
\item We show that finding an optimal subnetwork of an arbitrary graph is
NP-hard even when the one-sided degree restriction is $d = 1$. (Theorem \ref {HardThm}).
\end{itemize}

\section{Model}\label{sec:Model}

In this work, we are concerned with balanced bipartite graphs on $2n$ nodes.
In a balanced bipartite graph $G=(V,E)$, we have $V = L \cup R$, with $|L| = |R| = n$.
Our graphs will also have the following asymmetric degree restriction: all vertices
in $R$ have degree (exactly) $d > 0$.

On such a graph $G$, the following infection process occurs.  Each node $v$  becomes
infected independently at random with probability $\mu$ `by nature'.  Then, infected nodes
spread their infections independently to adjacent uninfected nodes with probability $p$.  As each new node becomes infected, they have one chance to infect their uninfected neighbors.  This is known as the independent cascade model in the literature \cite{kempe2003maximizing}.

Given that the above is a random process, we analyze the expected number of infected nodes
for a given choice of $n$, $\mu$, $p$, and graph $G$.
The goal of our work is to examine which networks among all bipartite graphs of a minimal degree on one side are most resilient to the spread of infections, i.e.\ which networks have the fewest infected nodes in expectation.  We also consider the computational hardness of determining the optimal subnetwork of one-sided minimal degree $d$ of an arbitrary bipartite graph. 

Blume~et~al.~\cite{blume2011network} study this model with respect to a cost/benefit analysis.  They consider strategic vertices who receive utility for each link formed but are penalized if they become infected.  They show a gap between the optimal graphs with respect to social welfare and graphs which satisfy conditions for strategic equilibria.  In this work we are solely concerned with socially optimal graphs and do not consider strategic behavior.  


One way to interpret the model and the quantity we are minimizing is with respect to independent edge percolation.  For a fixed graph $G$ on $n$ vertices, let each edge be present independently with probability $p$.  Let $|C(v)|$ denote the size of the (random) connected component containing the vertex $v$.  Then

\begin{equation}
\label{percEQ}
 I(G) := 1- \frac{1}{n} \E \left[ \sum_{v \in G} (1-\mu)^{|C(v)|} \right]
\end{equation}
is exactly the expected fraction of infected nodes in the $(\mu, p)$ model.  

Independent edge percolation on finite graphs is widely studied in probability and combinatorics.  If $G$ is the complete graph on $n$ vertices, the model is the \ER random graph.  Edge percolation on regular lattices is the topic of percolation theory in probability, and edge percolation on more general graphs has also been studied \cite{alon2004percolation,borgs2005random,nachmias2009mean}, but typically in the context of strong conditions (the `triangle condition', conditions on expansion) that ensure certain behavior at the phase transition.  

One topic in this field  that has not been considered in depth is extremal graphs with respect to percolation properties.  Network design to minimize the spread of infections is one example of such a problem, but many more can be imagined.  In fact, several other quantities can be interpreted with regard to the spread of infections.  For example, let the random variable
\begin{equation}
\label{SusEQ}
 S(G) = \frac{1}{n} \sum_v |C(v)|
\end{equation}
be the average component size of a graph after $p$-edge percolation.  This quantity, known as the susceptibility, is fundamental in the study of random graphs (eg. \cite{borgs2005random},\cite{spencer2007birth}).  It is not hard to show that the graph in a family of $n$-vertex graphs that minimizes $\E [S(G)]$ also minimizes the expected number of infected individuals in a single-origin model of infection in which one vertex at random is infected by nature, and then the infection spreads across edges with probability $p$.  

In a different model, that of general thresholds as studied in \cite{blume2011networks}, half-regular bipartite graphs are already extremely rich.  It can be shown that for $d=1$ every possible graph can be optimal under some choice of settings (Proposition \ref{genModelProp} in Section \ref{genModel}).

\section{Independent cascade on bipartite graphs}

As in the work of Blume~et~al.~\cite{blume2011networks}, we solve the problem of finding the optimal network satisfactorily for the smallest non-trivial degree bound ($d=1$ for half-regular bipartite graphs, $d=2$ for regular graphs), and for higher $d$ we exhibit two graphs that can be optimal.  


First we characterize the $d=1$ case, which is the simplest case for this model.  We first show that,
depending on the settings of $\mu$ and $p$, different graph structures become optimal.  Moreover, we can
characterize the set of optimal solutions -- namely, the network structure that minimizes $I(G)$, the expected fraction of infected nodes,
must always be a matching or a star.
Finally, we will point out that despite the optimality of one of the two extreme cases, there is non-monotonic behavior with respect to the size of the star.

\subsection{Half-regular graphs with $d$=1}

\begin{thm}
\label{d1thm}
For $d=1$, all $n$, and all settings of $\mu$ and $p$, either the perfect matching or an $n$-star (with $n-1$ isolated vertices)  minimizes $I(G)$.
\end{thm}
\begin{proof}


We observe that each feasible graph is a collection of stars with (possibly) some isolated vertices in $L$.  We therefore compute the expected fraction of infected individuals in the union of a $k$-star and $k-1$ isolated vertices, call this $\E [I_k]$:
\begin{equation}
\E [I_k] = \frac{L_k + (k-1) L_0 + k R_k}{2k},
\end{equation}
where $L_j$ is the probability that a vertex of degree $j$ in $L$ is infected, and $R_j$ is the probability that a vertex in $R$ joined to a vertex of degree $j$ is infected.  Note that the expected fraction of infected individuals in a perfect matching is exactly $\E[I_1]$ and the expected fraction in an $n$-star with $n-1$ isolated vertices in $L$ is $\E[I_n]$.  We will show that for $k \in [1, n]$, $\E[I_k]$ is minimized at either $k=1$ or $k=n$, and since any feasible graph is a union of stars, this shows that either the perfect matching or $n$-star is optimal.  

  We calculate
\begin{align*}
  L_j &= 1- (1- \mu) (1 - \mu p)^j 
\end{align*}
and
\begin{align*}
 R_j &= \mu + p - \mu p - (1-\mu)^2 p (1- \mu p)^{j-1},
\end{align*}
giving
\begin{align*}
\E [I_k] &= \frac{ 1- (1-\mu)(1-\mu)^k +(k-1)\mu  }{2k}   + \frac{\mu + p -\mu p -(1-\mu)^2 p (1-\mu p)^{k-1}}{2}
\end{align*}

Now define 
\begin{align*}
Q(k) &:= \frac{ 2(\E[I_k] - \E[I_1])}{1- \mu} + 2\mu p - p \\
&= \frac{1 - (1- \mu p)^k}{k} - (1-\mu)p(1-\mu p)^{k-1}  \\
&=\frac{1 - \alpha^k}{k} - \beta \alpha^{k}
\end{align*}
where we define $\alpha = 1- \mu p$ and $\beta = \frac{(1- \mu) p}{1- \mu p}$.

We will show that whenever $\frac{d Q}{d k} \ge 0$, $\frac{d^2 Q}{d k^2} <0$, which shows that $Q$ is a unimodal function of $k$ on the interval $[1,n]$ for any $n$, and in particular takes its minimum at one of its endpoints.  Because $Q$ is a linear function of $\E[I_k]$, this shows that $\E[I_k]$ takes its minimum at either $k=1$ or $k=n$.  We can assume $\mu \in (0,1)$ and $p >0$, since otherwise all $\E[I_k]$ is equal for all $k$.  

We compute

\begin{align*}
 \frac{d Q}{d k}&= - \frac{ (1- \alpha^k) + k (1 + \beta k) \alpha^k \log \alpha}{k^2}
\end{align*}

and

\begin{equation}
\label{d2Qeq}
\frac{d^2 Q}{d k^2} = \frac{2 (1- \alpha^k) +2 k \alpha^{k} \log (\alpha) - k^2 (1 + \beta k) \alpha^k \log^2 (\alpha)}{k^3}
\end{equation}
and so 
\begin{align*}
2 k^2 \cdot \frac{d Q}{d k} +k^3 \cdot \frac{d^2 Q}{d k^2} &= -\alpha^k k^2 \log (\alpha) (2 \beta + \log(\alpha) + \beta k \log(\alpha)  )
\end{align*}

Since $\log \alpha <0$, this is negative when $2 \beta + \log(\alpha) + \beta k \log(\alpha)$ is negative, i.e. when $k > - \frac{2}{\log \alpha} - \frac{1}{\beta}$, and so for such $k$ we have that whenever $\frac{d Q}{d k} \ge 0$, $\frac{d^2 Q}{d k^2} <0$.  If $- \frac{2}{\log \alpha} - \frac{1}{\beta} <1$, then we are done, since we need $Q$ to be unimodal on $[1,n]$.  

Otherwise, for $2 \beta + \log(\alpha) + \beta k \log(\alpha) \ge 0$, we show directly that $\frac{d^2 Q}{d k^2}$ is negative. From (\ref{d2Qeq}), we see that if 
\begin{align*}
H(k) &:=2 (1- \alpha^k) +2 k \alpha^{k} \log (\alpha) - k^2 (1 + \beta k) \alpha^k \log^2 (\alpha) < 0,
\end{align*}
 then $\frac{d^2 Q}{d k^2} <0$. We compute $H(0) =0$ and 
\begin{align*}
\frac{d H}{dk} &= - k^2 \alpha^k \log^2(\alpha) ( 3 \beta +\log (\alpha) + b k \log(\alpha))
\end{align*}
which is negative when $k>0$ and $3 \beta +\log (\alpha) + b k \log(\alpha) >0$, which is true by assumption for this range of $k$ since $\beta >0$.

\end{proof}

\begin{figure}[h!]
\centering
\includegraphics*[width=60mm]{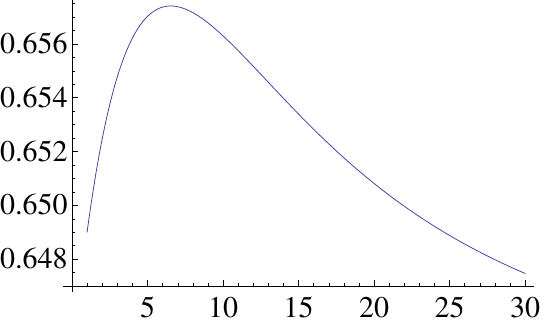}
\caption{Average infection probability as a function
of the degree of a star, for $\mu = 0.55$ and $p=0.4$.}
\label{fig:d1plot}
\end{figure}

Note that for $n$ large enough, the matching is better than the star if and only if $\mu < 1/2$.  However, there are already surprising effects
in the $d=1$ case -- for instance, while a star can be better than a matching,
a decomposition into smaller stars can be worse than either.    In Figure~\ref{fig:d1plot}, for the fixed parameters $\mu = .55, p=.4$, we plot the expected fraction of infected vertices in a $k$-star with $k-1$ isolated vertices for various values of $k$.


\subsection{Half-regular graphs with $d \ge 2$}

For $d \ge 2$, we first show two possibilities for optimal graphs.  We will prove the following proposition by solving appropriate extremal percolation problems:

\begin{thm}
\label{d2thm}
Both a collection of $K_{d,d}$'s and $K_{d,n}$ with $n-d$ isolated vertices can be optimal $d$-half-regular bipartite graphs. In particular,
\begin{enumerate}
 \item For any $p$ and any $d \ge 1$, for large enough $n$, there exists $\mu$ close enough to $1$ so that $K_{d,n}$ with $n-d$ isolated vertices is optimal.
\item For any $d$ and large enough $n$, there is a $\mu$ close enough to $0$, there exist $p$'s close enough to $0$ and to $1$ so that a collection of $K_{d,d}$'s is optimal.
\end{enumerate}
 
\end{thm}
\begin{proof} We prove the two parts separately:\\

\noindent \textit{1.}
If we set $\mu = 1 - n^{-2}$,  the RHS in Equation \ref{percEQ} becomes
\begin{equation}\label{eq:iso}
1 - n^{-3} \E [X_0(G)] +O(n^{-4}),
\end{equation}
 where $X_0(G)$ is the number of isolated vertices after $p$-edge percolation (each edge of the graph is deleted independently with probability $1-p$). So for large enough $n$, minimizing $I(G)$ becomes equivalent to maximizing the expected number of isolated vertices in a graph after $p$-edge percolation.  Since every vertex in $R$ has the same probability of being isolated due to the degree restriction, we wish to maximize the fraction of vertices in $L$ which are isolated.  The $K_{d,n}$ configuration has $n-d$ vertices which are isolated with probability $1$, and for $n$ large enough the contribution of the remaining $d$ vertices becomes negligible. \\
 
\noindent \textit{2.}
Set $\mu = n^{-2}$.  Then $I(G)$ in Equation \ref{percEQ} becomes $$n^{-3}\E \left[\sum_v |C(v)|\right] +O(n^{-3}),$$ and so minimizing $I(G)$ becomes equivalent to minimizing $\E [S(G)]$ from Equation \ref{SusEQ}. For $p=1$, we keep all the edges and so we need to minimize  $$\sum_v |C(v)| = \sum_C |C|^2 \le \sum_{C \in \mathcal C_R} |C|^2,$$
 where the first sum is over all vertices, the second over all components, and the third over all components containing a vertex in $R$. Since a collection of $K_{d,d}$'s has no isolated vertices in $L$, showing that such a graph minimizes $ \sum_{C \in \mathcal C_R} |C|^2$ suffices.  Considering all components containing a vertex in $R$, we note that each component has at least $d$ vertices from $L$, and the sum of the number of vertices from $R$ in all components equals $n$.  Under these conditions, minimizing with Lagrange multipliers gives each component of size $2d$, which is the $K_{d,d}$ configuration.

For $p \to 0$, set $\mu = n^{-3}, p=n^{-2}$.  A similar calculation to the above shows that minimizing $I(G)$ in this case is equivalent to minimizing $\sum_C |E(C)|^2$, where the sum is over all connected components and $|E(C)|$ is the number of edges in a component $C$.  Again we can relax the minimization since $K_{d,d}$'s will have no isolated $L$ vertices, and show that a collection of $K_{d,d}$'s minimizes $\sum_{C \in \mathcal C_R} |E(C)|^2$.  There are at most $n/d$ components in $\mathcal C_R$, and the total number of edges is $nd$.  Therefore $n/d$ components of $d^2$ edges each minimizes $\sum |E(C)|^2$, which completes the proof.  
\hfill$\Box$
\end{proof}

\begin{figure}[h]
\centering
\includegraphics*[width=40mm]{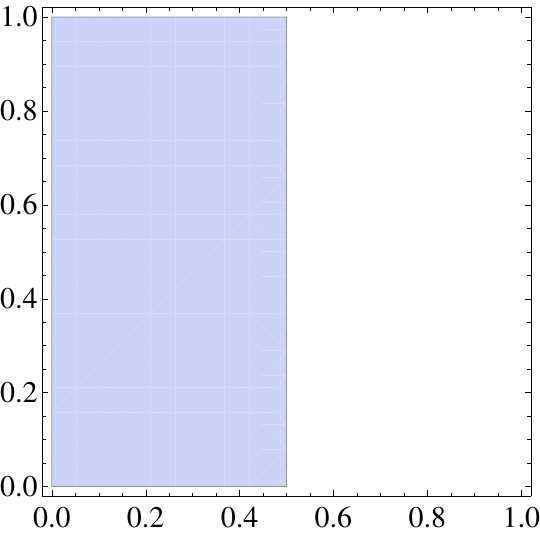}
\includegraphics*[width=40mm]{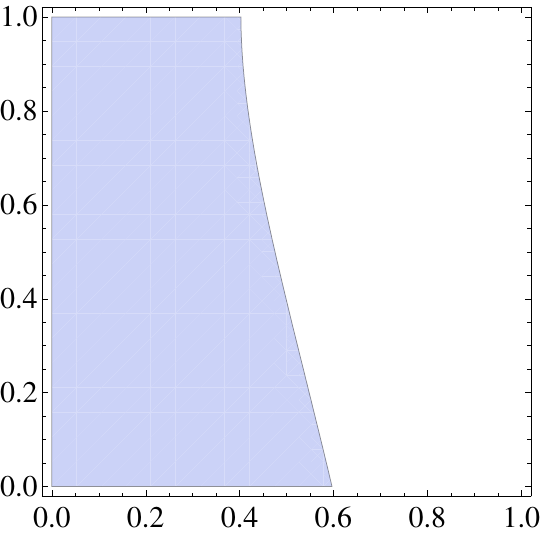}
\includegraphics*[width=40mm]{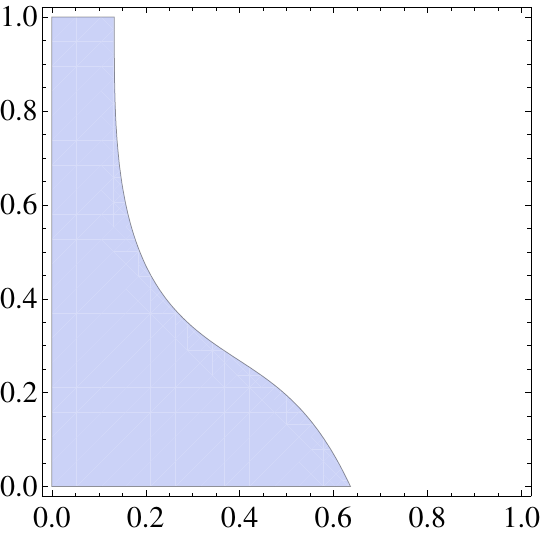}
\caption{The graphs are for $d=1$ (left), $d=2$ (center),
and $d=3$ (right), for $n \rightarrow \infty$. 
The $x$-axes are values of $\mu$, and the 
$y$-axes are values of $p$.  The colored
regions are where a $K_{d,d}$ decomposition
has a lower average infection
rate than $K_{d,n}$ with $n-d$ isolated vertices.}
\label{fig:dplots}
\end{figure}

In Figure~\ref{fig:dplots}, after solving the cases exactly, we indicate the regions in the parameter space for which $K_{d,d}$ and $K_{d,n}$ are better than one another in the large $n$ limit. It is straightforward to show that as $d \to \infty$, the cut-off for $p=1$ tends to $0$, and for $p \rightarrow 0$ the cut-off tends to $1$.

Given the results above, we might conjecture that for all $d \ge 1$ and $0 \le \mu, p \le 1$, either a $K_{d,d}$ decomposition or $K_{d,n}$ with $n-d$ isolated vertices would be the optimal $d$-half-regular, balanced bipartite graph on $2n$ vertices.  
Presently, however, we disprove such a conjecture.

\begin{thm}
\label{d3thm}
For $d=2$, there exist $2$-half-regular graphs on $2n$ nodes that are more resilient than either a $K_{2,2}$ decomposition or a $K_{2,n}$  with 
$n-2$ isolated vertices.  
\end{thm}

\begin{proof}
We take $n=4$ and consider the $2$-half-regular graph on $8$ vertices composed of a union of a $K_{3,2}$ and a $K_{1,2}$, with the degree requirement satisfied by the $3$ vertices on one side of the partition in the $K_{3,2}$ together with the $1$ vertex in the $K_{1,2}$. 

For the values $\mu = .302$ and $p = .801$, this graph is more resilient than either two copies of $K_{2,2}$ or the $K_{2,4}$ with two isolated vertices.  For these parameter settings, 
the average infection probabilities for the three graphs are approximately\footnote{We give approximate values to sufficient precision
to illustrate the difference in resilience.}  $.7197$, $.7207$, and $.7199$, respectively.  This counterexample graph was discovered via a careful computer search, 
 using Equation~\ref{percEQ}, over all half-regular graphs and a chosen set of settings for the parameters $\mu$ and $p$. 
\hfill$\Box$
\end{proof}

\subsection{A note on connected regular graphs}

We now briefly turn our attention back to the general model and consider what would happen if we disposed of any degree restriction, and instead
forced the graphs to be connected.  We show that with this different restriction, a similar phenomenon occurs as in the $d\ge2$ case, with optimally resilient graphs again
not lying on ``extremes."  Connected graphs are interesting in models where edges can be used for passing information, as well as disease.  There, finding connected
resilient graphs preserves the ability to spread information throughout the network while being as resilient as possible to the spread of disease.

If we try to find the optimally resilient connected graph for the $\mu,p$ model, we know that an optimal graph is always a tree, since any graph with cycles can have
an edge removed without hurting resiliency.  It is also interesting to note that, because of this, connectivity naturally gives us a different restriction on bipartite graphs than
half-regularity.

A connectivity requirement is somewhat different than the regular or half-regular case.  
For example, $K_{d,d}$ decompositions, which are sometimes optimal in the half-regular case, are
no longer allowed if the graph must be connected.  Similarly, for $d$-regular graphs, Blume~et~al.~\cite{blume2011networks} show that the optimal $2$-regular finite graph on $3n$ nodes is always 
a triangle decomposition; this is again not connected.

It is then natural to begin by considering the path and the star graphs.\footnote{We note that Blume~et~al.~\cite{blume2011networks} show that the infinite path
can be the optimal $2$-regular graph.}  In the case of infinite graphs, it is easy to exactly find the expected infection probability of both the infinite star and the
infinite path.  For the case of the infinite star, we can assume the center is infected (as long as $\mu, p$ are constants $> 0$), and therefore the probability of infection for a leaf is simply
\begin{equation}\label{infstar}
\mu + (1-\mu)p.
\end{equation}
In the case of the infinite path, Equation~\ref{percEQ} gives  an average infection rate of 
\begin{equation}\label{infpath}
\sum_{i=1}^{\infty}{i (1-(1-\mu)^i) p^i(1-p)^2  } = \frac{\mu p-\mu p^3 + \mu^2p^3}{p(1-p+\mu p)^2}.
\end{equation}
It is also easy to see that the quantities in Expressions~\ref{infstar} and~\ref{infpath} are upper bounds for finite stars and paths, respectively, yet either of
these can be optimal depending on the settings of $\mu$ and $p$.

The natural question again arises whether a star or a path must always be the most resilient graph, and the answer is, perhaps by now, unsurprisingly, no.

For $n=5$, we compare the $5$-path to the star on $5$ nodes to a $5$-node ``fork graph" (Figure~\ref{fig:fork}), and we show that a fork graph
can be more resilient than either one of the two ``extremes." For the values $\mu = .63$ and $p = .7$, the average infection probabilities for the star, path, and 
fork graphs are approximately .8906, .8907, and .8905, respectively.  Figure~\ref{f:pathstarfork}, computed from plotting the exact infection rates
on the three graphs shows the narrow region where the fork is more resilient than the other two extreme graphs.

\begin{figure}[h!]
\centering
\includegraphics*[width=110mm]{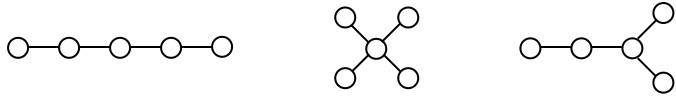}
\caption{Left to right: the path, star, and fork graphs on $5$ nodes. These graphs comprise all the trees on $5$ nodes, up to isomorphism. Hence, the most resilient
$5$-node connected graph must come from this set of graphs, $\forall\ 0 \le u,p \le 1$.}
\label{fig:fork}
\end{figure}

\begin{figure}[t!]\label{f:pathstarfork}
\centering
\includegraphics*[width=50mm]{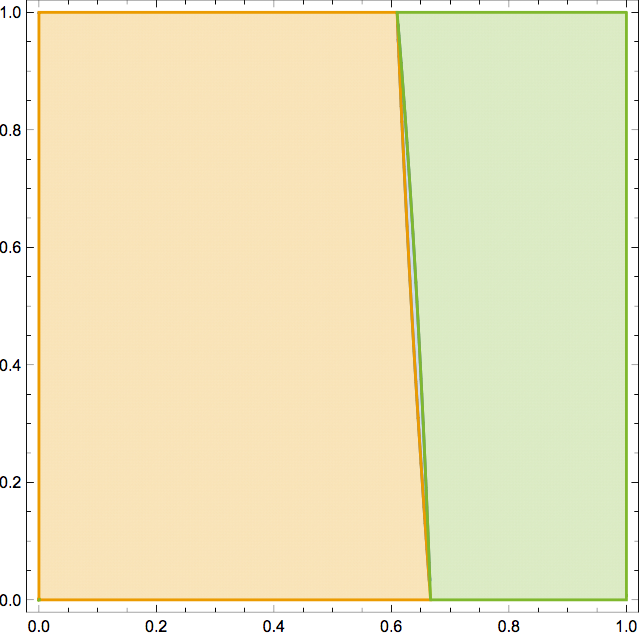}
\caption{The orange region is where the $5$-path is the most resilient $5$-node connected graph; 
the green region is where the star on $5$ nodes is the most resilient $5$-node connected graph; the
small blue region in the center is where the fork is the most resilient $5$-node connected graph.  $\mu$ runs
along the horizontal axis and $p$ runs on the vertical axis.}
\label{fig:fork}
\end{figure}

%

\section{Optimal subnetworks of arbitrary graphs}
In this section we consider the problem of finding an optimal bipartite subnetwork of arbitrary bipartite graphs.

Let $G=(V,E)$ be a bipartite graph with $V = L \cup R$ with degree $\ge d$ for vertices in $R$.
We call  the problem of finding a subgraph of $G$, $G' = (V,E')$, with minimum degree $d$ for
vertices in $R$, as to minimize $I(G')$, the \textbf{optimal bipartite subnetwork problem}.

\begin{thm}
\label{HardThm}
For all $d \ge 1$ the optimal bipartite subnetwork problem is NP-hard.
\end{thm}
\begin{proof}
For $d=1$ we reduce from exact set cover.  An instance of exact set cover is a family of subsets $\mathcal{F}$ of a ground set $U$.
The goal is to find a subcollection of sets $\mathcal{F}' \subseteq F$ such that each element in $U$ appears in exactly one set in $F'$.
This problem is NP-hard~\cite{Karp72}. We will assume w.l.o.g. that all sets in $F$ are the 
same size, $k$ (we can append new elements to smaller sets).

For our reduction, we construct an instance of the optimal bipartite subnetwork problem as follows.
The graph $G$ will contain vertices $L \cup R$, with $R = U$ and $L = \mathcal{F}$. We form an edge
$(l,r) \in E$, where $l \in \mathcal{F}$ and $r \in U$ if $r \in l$.
Applying Equation~\ref{eq:iso}, there is a setting of $\mu$ and $p$ such that the optimal network
will maximize the number of isolated vertices, subject to our constraints.

It is clear that if an exact cover exists, there will be subgraph of $G$ with $|\mathcal{F}| - |U|/k$ isolated vertices --
namely the one that uses all edges from the cover.  On the other hand, if there is no exact cover, the number
of isolated vertices will be $\le |\mathcal{F}| - |U|/k -1$.

For $d=2$, we use Theorem~\ref{d2thm}, part 2, that there exist settings for $\mu$ and $p$ such that a
$K_{d,d}$ decomposition is optimal in any graph if it exists.  The problem of decomposing a bipartite
graph into vertex-disjoint $K_{2,2}$ is NP-hard~\cite{FederM95}

For $d \ge 3$ we reduce from the problem of finding a $d$-clique decomposition of an arbitrary graph, known to 
be NP-hard~\cite{KirkpatrickH78}.
An instance of a $d$-clique decomposition problem is a graph $G=(V,E)$ and a solution is a partition of $G$
into vertex-disjoint $d$-cliques.

For our reduction we make a bipartite graph $\hat{G}=(\hat{V},\hat{E})$ with $\hat{V} = L \cup R$ and 
$|L|=|R| = |V|$ and $(l_i, r_j) \in \hat{E}$ if $(v_i,v_j) \in E$ or $i = j$.
Again, by Theorem~\ref{d2thm}, part 2, there exist settings for $\mu$ and $p$ such that a
$K_{d,d}$ decomposition is optimal.
 Such a decomposition will exist in our case if and only if the original graph $G$ had a $d$-clique
 decomposition.
\hfill$\Box$
\end{proof}

\section{General threshold model}
\label{genModel}

Blume~et~al.~\cite{blume2011networks} consider a generalization of the $(\mu,p)$ model which we will call the general threshold model. In this model, each vertex is assigned a non-negative integer $i$ which represents the number of infected neighbors required to infect that vertex.  If $i=0$, the vertex is infected `by nature'.  We assign these integers randomly and independently according to some common distribution, where $\Pr[ i] =: \mu_i$, and $\sum \mu_i =1$.  The sequence $\{ \mu_i\}$ comprises 
the parameters for the model.  The $\mu, p$ model is a special case of the cascade model with
\begin{equation*}
\mu_i = 
\begin{cases}
\mu & \mbox{if } i = 0 \\
\mu_i = (1 - \mu) p (1-p)^{i-1} & \mbox{if } i \ge 1.  
\end{cases}
\end{equation*}

In the case of $d$ regular graphs, \cite{blume2011networks} shows that for $d=2$, the optimal graphs\footnote{Their choice of objective function is slightly different: they minimize the maximum probability of infection over all vertices.} are collections of disjoint triangles or the $n$-cycle.  For $d \ge 3$, they show that both collections of disjoint $(d+1)$-cliques and the infinite $d$-regular tree can be optimal, but there are choices of parameters for which neither is optimal.  

For half-regular bipartite graphs, already the case $d=1$ shows the richness of this model: each $k$-star can be optimal under some choice of parameters:

\begin{prop}
\label{genModelProp}
 For every $k \ge 1$ there exists $\eps$ small enough so that for the choice of parameters $\mu_0 =.6$, $\mu_1 =\eps$, and $\mu_{k+1} = .4 - \eps$ in the general threshold model, the $k$-star is the optimal $1-$half-regular bipartite graph.  
\end{prop}

\begin{proof}
Set the parameters of the general threshold model as above.  For $j\le k$, the expected fraction of infected individuals in a $j$-star with $j-1$ isolated vertices is:
\begin{align*}
\E [I_j] &= \frac{1}{2j} \left [  .6 \cdot 2j + \eps(1-.4^j) + .6 \eps j + O(j \eps^2) \right ] \\
&= .6 + .3 \eps + \frac{1 -.4^j}{2j} \eps + O(\eps^2)
\end{align*}
The function $ \frac{1 -.4^j}{2j}$ is a strictly decreasing function of $j$, so for small enough $\eps$ the $k$-star is better than any $j$-star with $j < k$.  
And for $j>k$,
\begin{align*}
\E [I_j] &\ge .6 + .3 \eps +  \frac{1 -.4^j}{2j} \eps    + \frac{ (.4 -\eps) q_{j,k+1}}{2j} +\frac{ \eps (.4 -\eps)}{2j} \sum_{i=k+1}^{j-1} (j-1) p_{j,i} 
\end{align*}
where $ q_{j,k+1} = \Pr[Bin(j,.6)\ge k+1]$ and $p_{j,i} = \Pr[Bin(j,.6)=i]$.  For $j\le 2 k$, and $\eps $ sufficiently small, $ \frac{ (.4 -\eps) q_{j,k+1}}{2j} > \frac{1 -.4^j}{2j} \eps$ and so $\E [I_j] > \E [I_k]$.  For $j> 2k$, the term $\frac{ \eps (.4 -\eps)}{2j} \sum_{i=k+1}^{j-1} (j-1) p_{j,i} $ is bounded below by $\eps$ times a constant independent of $j$, and so the $k$-star is optimal.  
\hfill$\Box$
\end{proof}









\bibliographystyle{acm}
\bibliography{paper}

\end{document}